\RequirePackage{fix-cm}
\documentclass[smallcondensed]{svjour3}     
\smartqed  
\usepackage{graphicx}
\usepackage{epsfig}
\usepackage{framed}
\usepackage{amssymb}
\usepackage{amsmath}
\usepackage{latexsym,epic,eepic}
\usepackage{fancybox}
\usepackage{color}
\usepackage{algorithmic}
\usepackage{algorithm}

\newcommand\anthony[1]{{\color{blue}
[#1 - \textbf{Anthony}]}}

\newenvironment{mybox}{
  \framed{\advance\hsize-\width}%
}{%
  \endframed
}

\newcommand{\OMIT}[1]{}

\newcommand{\N}{\mathbb{N}}
\newcommand{\Z}{\mathbb{Z}}
\newcommand{\Q}{\mathbb{Q}}
\newcommand{\vecU}{\ensuremath{\text{\bfseries u}}}
\newcommand{\vecV}{\ensuremath{\text{\bfseries v}}}
\newcommand{\vecW}{\ensuremath{\text{\bfseries w}}}
\newcommand{\Dom}{\ensuremath{\text{\scshape Dom}}}
\newcommand{\Size}[1]{\ensuremath{\|#1\|}}

\newcommand{\floor}[1]{\ensuremath{\lfloor #1 \rfloor}}
\newcommand{\Ord}{\ensuremath{\text{ord}}}
\newcommand{\sem}[1]{[\![#1 ]\!]}

\newcommand{\defn}[1]{\emph{#1}}

\newcommand{\rightshift}{\ensuremath{\text{\scshape RS}}}

\newcommand{\StirlingFirst}[2]{\ensuremath{\begin{bmatrix}
                                                #1 \\
                                                #2
                                            \end{bmatrix}}}

\newcommand{\Identity}{\ensuremath{\text{Id}}}

\newcommand{\false}{\ensuremath{\text{\scshape false}}}

\newcommand{\NP}{\ensuremath{\text{\scshape NP}}}
\renewcommand{\P}{\ensuremath{\text{\scshape P}}}
\newcommand{\coAM}{\ensuremath{\text{\scshape coAM}}}

\newcommand{\Stab}{\ensuremath{\text{Stab}}}
\newcommand{\permgroup}{\ensuremath{\mathcal{S}}}

\begin{document}

\title{
A linear-time algorithm for the orbit problem over cyclic groups
}


\author{Anthony W. Lin \and
        Sanming Zhou
}


\institute{A.~W. Lin\at
              Yale-NUS College \\
              10 College Ave West, Singapore 138609 \\
              Tel.: +65 6601-3699 \\
              \email{anthony.w.lin@yale-nus.edu.sg}           
           \and
           S. Zhou \at
            School of Mathematics and Statistics \\ 
            University of Melbourne \\
            Parkville, Victoria 3010, Australia.
}

\date{Received: date / Accepted: date}

\maketitle

\begin{abstract}
The orbit problem is at the heart of symmetry reduction methods for model
checking concurrent systems. It asks whether two given
configurations in a concurrent system (represented as finite strings over some
finite alphabet) are in the same orbit with respect to a given finite 
permutation group (represented by their generators) acting on this set of 
configurations by permuting indices. 
It is known that the problem is in general as hard as the graph isomorphism
problem, whose precise complexity (whether it is solvable in polynomial-time)
is a long-standing open problem.
In this paper, we consider the restriction of the orbit problem 
when the permutation group is cyclic (i.e. generated by a single
permutation), 
an important restriction of the problem. It is known that this subproblem
is solvable in polynomial-time. Our main result is a linear-time algorithm for 
this subproblem.

\OMIT{To complement
this result, we delineate the boundary of tractability by showing that 
permutation groups generated by two permutations already suffice to make the orbit
problem as hard as the graph isomorphism problem.
, though polynomial-time 
solvability can be retained for abelian permutation groups with a fixed number of 
generators.}
\keywords{symmetry reductions \and model checking \and cyclic groups
    \and orbits}
\end{abstract}

\section{Introduction}
\label{sec:intro}

Since the inception of 
model checking (cf. \cite{Birth}), a key challenge in verifying
 concurrent systems has always been how to circumvent the state explosion problem,
which is exponential in the number of processes and in the number of finite-domain 
variables. The fundamental algorithmic problem can essentially be construed as a 
reachability problem in an exponentially-sized graph that is succinctly 
represented (e.g. in some concurrent programming language). Among others, 
symmetry reduction \cite{ID96,CJEF96,ES96} has emerged to be an effective technique 
in combatting the state explosion problem. The essence of symmetry reduction
is to identify symmetries in the system 
and avoid exploring states that are ``similar'' (under these symmetries) to 
previously explored states, thereby speeding up model checking.
\OMIT{
In the case of concurrent systems, the types of symmetries that have received a 
lot of attentions include \defn{process symmetries}, which arise due to 
replications of concurrent components. In this paper, we shall consider only
such symmetries. 
}

Every symmetry reduction method has to deal with the following two computationally
difficult problems: (1) 
how to identify symmetries in the given system, and (2) how to check that
two configurations are similar under these symmetries. To simplify our discussion
of Problem 1, we will restrict our discussion to \emph{process symmetries}. 
[Extensions 
    to 
\emph{data symmetries} are possible, e.g., see the recent result of \cite{SunJun13},
which gives a general reduction of process and data symmetry identification in 
concurrent systems to symmetry identification in the solutions to constraints in 
the constraint-satisfaction problem.] In this case, 
for concurrent systems with $n$ processes, Problem 1 
amounts to searching for a group $G$ of permutations on $[n] := 
\{1,\ldots,n\}$  such that the system behaves in an identical
way under the action of permuting the indices of the processes by any $\pi
\in G$. For example, for a distributed protocol with a ring topology, the group 
$G$ could be a \emph{rotation group} generated by the ``cyclical right shift'' 
permutation $\rightshift$ that maps $i \mapsto i+1 \mod n$ for each $i\in[n]$. 
See Example \ref{ex:token} for concrete examples.
Although Problem 1 is computationally hard in general, a lot of research advances 
has been made in the past decade (e.g. see the recent survey \cite{WD10}, and also 
the recent paper \cite{SunJun13} for a more general
technique that covers both process and data symmetries).
Now the group 
$G$ partitions the state space
of the concurrent system (i.e. $\Gamma^n$ for some finite set $\Gamma$) into
equivalence classes called \defn{($G$-)orbits}. Problem 2 is 
essentially the \defn{orbit problem (over finite permutation groups)}: 
given $G$ and two configurations $\vecV, \vecW \in \Gamma^n$, determine 
whether $\vecV$ and $\vecW$ are in the same $G$-orbit.
For example, if $G$ is generated by $\rightshift$ with $n = 4$,
the two configurations $(1,1,0,0)$ and $(0,0,1,1)$ are in the same orbit.
These two computational problems can be studied independently. The focus of this 
paper is the second problem, i.e., the orbit problem.

\begin{mybox}
\begin{example}
    \label{ex:token}
    \textbf{Two token-passing protocols with multiple tokens:} These examples are
nondeterministic versions of the randomised self-stablising protocol of Israeli and 
Jalfon \cite{IJ90} (also see \cite{Norman-survey}). 
\medskip

\begin{center}
    \epsfig{file=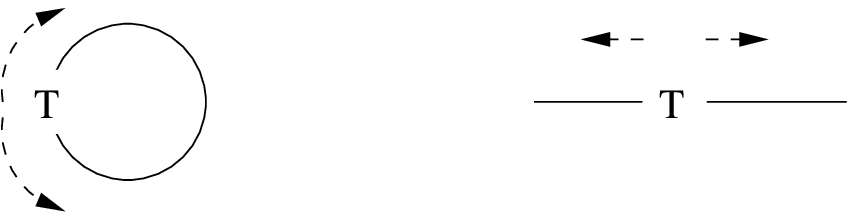}
\end{center}

\smallskip
In the first example (left figure), there are
$n$ processes $P_1,\ldots,P_n$ connected in a ring-shaped topology (i.e. the
neighbours of $P_i$ are $P_{i+1 \mod n}$ and $P_{i-1 \mod n}$). There are $m \leq n$
tokens in the network, each held by a unique process. At any given step, a unique
process $P_i$ holding a token is nondeterministically chosen by a scheduler and is
permitted to pass the token to its neighbour $P_j$ (i.e. either the left 
$P_{i-1 \mod n}$ neighbour or the right neightbour $P_{i+1 \mod n}$). If $P_j$ 
already had a token, it will simply merge the two tokens into one, which reduces
the total number of tokens in the network by 1. For each number $n$ of processes,
this description yields a transition system. For example, configurations of the 
system
are of the form $(\alpha_1,\ldots,\alpha_n) \in \{\bot,\top\}^n$, where $\top$ 
(resp. $\bot$) constitutes that the process holds (resp. does not hold) a token.
The symmetry group $G_n$ of the transition system is the \emph{Dihedral group}
$D_n$, which is generated by the cyclical right shift RS ($i \mapsto i+1 \mod n$)
and the \emph{reflection} ($i \mapsto n-i$, for each $i \in \{1,\ldots,n\}$).
In the standard composition of disjoint cycles notation, these permutations can be 
written
as $(1,2,\ldots,n)$ and $(1,n)(2,n-1)\cdots (\lfloor n/2\rfloor,\lceil n/2\rceil)$,
respectively.

In the second example (right figure above), we modify the first example by 
disconnecting the line
between $P_1$ and $P_n$, which results in a line-shaped topology. In effect, $P_1$
(resp. $P_n$) can only pass a token to $P_2$ (resp. $P_{n-1}$). The symmetry group
$G_n'$ of the system in this case is the group generated by the reflection mapping
that maps $i \mapsto n-i$, for each $i \in \{1,\ldots,n\}$.
\end{example}
\end{mybox}

The orbit problem (OP) was first studied in the 
context of model checking by
Clarke \emph{et al.} \cite{CJEF96} in which it was shown to be in $\NP$ but is
as hard as the graph isomorphism problem, whose precise complexity (whether it 
is solvable in polynomial-time) is a long-standing open problem.
The difficulty of the problem is due to the fact that 
the input group $G$ is represented by a set $S$ of generators and that the size 
of $G$ can be exponential in $|S|$ in the worst case. 
There is also a closely related variant 
of OP called the \defn{constructive orbit problem (COP)}, which
asks to compute the lexicographically smallest element $\vecW \in \Gamma^n$ 
in the orbit of a given configuration $\vecV \in \Gamma^n$ with respect to 
a given group $G$. OP is easily reducible to COP, though the reverse direction 
is by no means clear. COP was initially studied in the context of graph 
canonisation by Babai and Luks \cite{BL83}, in which COP was 
shown to be $\NP$-hard (in contrast, OP is unlikely to be $\NP$-hard since it would
imply the collapse of the polynomial-time hierarchy to its second 
level\footnote{For, if it were NP-hard, then the \emph{coset intersection 
problem (for permutation groups)} would be NP-hard, owing to its 
polynomial-time equivalence to the orbit problem
\cite{CEJS98}. By the well-known results of \cite{BM88,GMW86} (also see
\cite[Section 6.5, Chapter 27]{HandbookCombinatorics}), which essentially shows
that the coset intersection problem is in the complexity class $\coAM$ 
(contained in the second level of the polynomial hierarchy), this would mean
that the polynomial hierarchy collapses to the second level. 
}).
In the context of model checking, COP was first 
studied by Clarke \emph{et al.} \cite{CEJS98}, in which a number of ``easy 
groups'' for which COP becomes solvable in $\P$ are given including 
polynomial-sized groups (e.g. rotation groups), the full symmetry group
$\permgroup_n$ (i.e. containing all permutations on $[n]$), 
and disjoint/wreath products of easy groups (cf. \cite{DM09}).

In this paper, we consider the orbit problem over \defn{cyclic 
groups} (i.e. generated by a single permutation $\pi \in \permgroup_n$), which
is an important subproblem of OP. 
In the case of rotation groups,
one can do a simple enumeration of the group elements and solve the orbit 
problem in polynomial-time. [More precisely, if the group has $m$ elements,
this algorithm runs in time $O(mn)$, which is already quadratic over rotation
groups.] 
However, cyclic subgroups of $\permgroup_n$ can even be of size exponential
in $n$ (see Proposition \ref{prop:exp} below), which rules out this enumeration
strategy. It turns out that the orbit problem over cyclic groups is known to
be solvable in polynomial-time (e.g. see \cite{BL83,Luks93}, where this is
shown for a much larger class of permutation groups denoted as $\Gamma_d$
for every \emph{fixed} $d$, which contains solvable groups).
\OMIT{
It is also possible to
combine cyclic groups with other easy groups from \cite{CEJS98} via 
disjoint/wreath product operators so as to efficiently solve OP for more complex 
permutation groups. 
}
\OMIT{
Secondly, it subsumes a commonly
occurring class of symmetries for concurrent systems: the rotation groups. 
Unlike the case of rotation groups however, the 
size of cyclic groups can be exponential in $n$ (see Proposition \ref{prop:exp} 
below), which rules out a naive enumeration of the group elements. 
}
Another way to see that OP over cyclic groups is solvable in polynomial-time
is by a quadratic-time reduction to 
the classical orbit problem over rational matrices \cite{KL86}: given 
a rational $n$-by-$n$ matrix $M$ and two rational vectors 
$\vecV, \vecW \in \Q^n$, determine if there exists $k \in \N$ such that $M^k 
\vecV = \vecW$. In fact, the two problems coincide when $M$ is restricted to 
\defn{permutation matrices} \cite{Brualdi}, i.e., 0-1 matrices with precisely
one column for each row with entry 1.
\OMIT{
To see this, given a permutation $\pi$ on $[n]$, simply take an 
$n$-by-$n$ 0-1 matrix $A = (A[i,j])_{1\leq i,j \leq n}$ such that 
$A[i,j] = 1$ iff $\pi(j)=i$. The reverse direction is similar.
}
That OP over cyclic groups is in $\P$ follows from Kannan and Lipton's 
celebrated result \cite{KL86} that OP over rational matrices is in $\P$.

\OMIT{
It turns out that the above algorithm for the orbit problem obtained by a 
reduction
to OP over rational matrices runs in time at least $O(n^4)$.
}
Mere polynomial time-complexity is far from sufficient for the purpose of 
symmetry reduction methods, since a model checker will 
have to invoke an algorithm 
for the orbit problem \emph{once each time a new configuration} in the given 
transition system is visited (e.g. see \cite{WD10}). Recent case studies in 
\cite{SunJun13} suggest that the cost of
solving the orbit problem often becomes extremely prohibitive, even more so than 
the cost of computing the symmetries\footnote{Some examples in \cite{SunJun13} (even
with a small number of processes) require a model checker to invoke an algorithm for
the orbit problem hundreds to thousands of times for one transition system.}.
Therefore, lightweight methods for dealing with the orbit problem are 
crucial for the success of symmetry reductions in model checking. 




\OMIT{
Symmetry reduction 

Symmetry extraction and The orbit problem (in general). 
}

\noindent
\textbf{Contributions.} In this paper, we provide an algorithm for
the orbit problem over cyclic groups that runs in linear-time.
To this end, we provide a linear-time reduction to the problem of 
solvability of systems of linear congruence equations. 
The reduction exploits subtle connections to the string searching problem. 

As for the solvability of systems of linear congruence equations, there is
a well-known algorithm (based on the extended Euclidean algorithm)
that runs in linear-time assuming constant-time integer arithmetic operations.
However, when we measure the number of bit operations (i.e. bit complexity 
model), it turns out that the algorithm runs in time 
cubic in the number of equations in the systems. To address this issue, we
restrict the problem to input instances provided by our reduction from the orbit
problem. 
We offer two solutions.
Firstly, we show that the average-case complexity of the algorithm under the bit
complexity model is $O(\log^5 n)$, which is sublinear [Here, $n$ measures the
size of the input to the orbit problem.] Secondly, we provide another algorithm 
that uses at most linearly many bit operations \emph{in the 
worst case} (though on average it is worse than the first algorithm).

It turns out that
permutation groups generated by two permutations already suffice to make the 
orbit
problem as hard as the graph isomorphism problem. This is almost a direct
corollary of the polynomial-time reduction in \cite{CJEF96}
from the graph isomorphism problem to the orbit problem over some group $G$.
It turns out the group $G$ that is produced by the reduction of \cite{CJEF96}
is \emph{not} any arbitrary group and could easily be generated by
two generators (for the same reason that the full symmetry group $\permgroup_n$ 
on $\{1,\ldots,n\}$ can be generated by the permutations $(1,2)$ and 
$(1,2,\ldots,n)$).


\OMIT{
\noindent
\textbf{Related Work.} ...
Furthermore, a recent benchmarking from \cite{SunJun13} shows that generators that
are more complex than rotation groups (i.e. permutations which are products of many 
disjoint cycles) frequently appear in the symmetry group of concurrent systems\anthony{Provide a specific example}.
\smallskip
}

\smallskip
\noindent
\textbf{Organisation.} Section \ref{sec:prelim} contains basic definitions,
notations, and results that will be used throughout the rest of the paper.
We provide our linear-time reduction from the orbit problem to equations solving
in Section \ref{sec:reduce} (Algorithm \ref{algo:main}). Thus far, we assume 
that arithmetic operations take
constant time. We deal with the issue of bit complexity in Section
\ref{sec:bit}. 
We conclude with future work in Section \ref{sec:conc}.
\smallskip

\noindent
\textbf{Acknowledgment.} We thank the anonymous referees of the conference version
for their helpful feedback. Lin was supported by Yale-NUS Startup Grant; part of 
the work was done when Lin was at Oxford supported by EPSRC (H026878). Zhou was 
supported by ARC (FT110100629).

\section{Preliminaries}
\label{sec:prelim}
\noindent
\textbf{General Notations}: We use $\log$ (resp. $\ln$) to denote logarithm 
in base 2 (resp. natural logarithm). We use the standard interval 
notations
to denote a subset of integers within that interval. For example, $[i,j)$
denotes the set $\{ k \in \Z : i \leq k < j \}$. Likewise, for each positive 
integer $n$, we use $[n]$ to denote the set $\{1,\ldots,n\}$. We shall also
extend arithmetic operations to sets of numbers in the usual way: whenever
$S_1, S_2 \subseteq \Z$, we define $S_1 + S_2 := \{ s_1 + s_2 : s_1 \in S_1,
s_2 \in S_2\}$ and $S_1S_2 :=  \{ s_1 \times s_2 : s_1 \in S_1, s_2 \in S_2\}$.
In the context of arithmetic over $2^{\Z}$, we will treat a number $n \in \N$
as the singleton set $\{n\}$. That way, for $a,b \in \N$, the notation $a+b\Z$
refers to the \defn{arithmetic progression} $\{ a + bc : c \in \Z \}$, where
$a$ (resp. $b$) is called the \defn{offset} (resp. \defn{period}) of
the arithmetic progression.
Likewise, for a subset $S \subseteq \N$, we use $\gcd(S)$ to denote the 
greatest common divisor of $S$. 

We will use standard notations from formal language theory. Let $\Gamma$ be
an \defn{alphabet} whose elements are called \defn{letters}.
A word (or a string) $w$ over $\Gamma$ is a finite sequence of elements
from $\Gamma$. We use $\Gamma^*$ to denote the set of all words over $\Gamma$. 
The length of $w$ is denoted by $|w|$. Given a word $w = a_1
\ldots a_n$, the notation $w[i,j]$ denotes the subword $a_i\ldots a_j$. For
a sequence $\sigma = i_1,\ldots,i_k \in [n]^*$ of \emph{distinct} indices of 
$w$, 
we write $w[\sigma]$ to denote the word $a_{i_1}\ldots a_{i_k}$. We also define
$\rightshift(w)$ to be $a_na_1a_2\ldots a_{n-1}$, i.e., the word $w$ cyclically 
right-shifted. 
\smallskip

\noindent
\textbf{Number Theory:} In the sequel, we will use some standard results
in number theory and algorithmic number theory. The first result is 
Linear Congruence Theorem and its application to solving a system of linear
congruences. The second result is Chinese Remainder Theorem.

Linear Congruence Theorem (e.g. see \cite[Chapter 31.4]{Cormen} or 
\cite[Theorem 4.5]{Shoup}) gives a fast method of determining whether an
equation of the form $ax \equiv b\pmod{n}$ is solvable and, whenever
it is solvable, the set of solutions to $x$.
\begin{lemma}[Linear Congruence Theorem]
The equation $ax \equiv b \pmod{n}$ is solvable for the unknown $x$ iff 
$d | b$, where $d = \gcd(a,n)$. Furthermore, if it is solvable,
then the set of solutions equals $x_0 + (n/d)\Z$, for some $x_0 \in [0,n/d)$
    that can be computed in time $O(\log n)$ (assuming constant-time
    arithmetic operations).
\label{lm:linear-congruence}
\end{lemma}
An immediate application of Linear Congruence Theorem is to determine the set
of solutions to a \emph{system} of linear congruences. A \emph{system of linear congruence equations} is a relation of the form $\bigwedge_{i=1}^m x \equiv 
a_i \pmod{b_i}$. [In general, a system of linear congruence equations might
    take an equation of the form $ax \equiv b \pmod{n}$, but we do not need
this general form in the sequel.]
The set of solutions $x \in \Z$
to this system is denoted by $\sem{\bigwedge_{i=1}^m x \equiv a_i 
\pmod{b_i}}$, which equals $\bigcap_{i=1}^m \left(a_i + b_i\Z\right)$.
The system is \defn{solvable} if the solution set
is nonempty. We use $\false$ to denote $x \equiv 0\pmod{2} \wedge
x \equiv 1 \pmod{2}$, which is not solvable.
\begin{proposition}
For any solvable system of linear congruence equations 
$\varphi(x) := \bigwedge_{i=1}^m x \equiv a_i \pmod{b_i}$, we have
$\sem{\varphi(x)} = \sem{x \equiv a \pmod{b}}$
for some $a,b \in \Z$. Furthermore, there exists an 
algorithm which computes $a,b$ in linear time (assuming constant-time 
arithmetic operations).
\label{prop:sys-eqs}
\end{proposition}
Proposition \ref{prop:sys-eqs} is witnessed by Algorithm \ref{algo:sys-eqs},
which is simply a repeated application of Linear Congruence Theorem.
\begin{algorithm}
\caption{Solving a system of modular arithmetic equations\label{algo:sys-eqs}}
\begin{algorithmic}
\REQUIRE A system of modular arithmetic equations $\bigwedge_{i=1}^m 
    x \equiv a_i \pmod{b_i}$
\ENSURE Solution set $\sem{\bigwedge_{i=1}^m x \equiv a_i\pmod{b_i}}$ as
    $\emptyset$ or an arithmetic progression $a + b\Z$.
\STATE $a := 0$; $b := 1$;
\FOR{$i=1,\ldots,m$}
    \STATE $\varphi(y) := by \equiv a_i - a \pmod{b_i}$;
    \STATE Apply algorithm from Lemma \ref{lm:linear-congruence} on
        $\varphi$ returning either $\emptyset$ or $a'+b'\Z$ for 
        $\sem{\varphi}$;
    \STATE \textbf{if} $\sem{\varphi} = \emptyset$ \textbf{then} 
        \textbf{return} NO \textbf{else} $a := a'b+a$; $b := bb'$ 
            \textbf{end if}
\ENDFOR
\RETURN $a+b\Z$;
\end{algorithmic}
\end{algorithm}


\begin{remark}
The number of bits that is used to maintain $a$ and $b$ in the worst case
is linear in the size $\sum_{j=1}^m (\log a_j + \log b_j)$ of the input. 
This justifies treating a single arithmetic operation as a constant-time
operation. We will discuss bit complexity in Section \ref{sec:bit}.
\end{remark}

In the sequel, we will also use Chinese Remainder Theorem (e.g. see
\cite[Section 31.5]{Cormen} or \cite[Theorem 2.6]{Shoup}).
\begin{proposition}[Chinese Remainder Theorem]
Let $n_1,\ldots,n_k$ be pairwise relatively prime positive integers, and
$n = \prod_{i=1}^k n_i$. The ring $\Z_n$ and the direct product of rings
$\Z_{n_1} \times \cdots \times \Z_{n_k}$ are isomorphic under the function
$\sigma: \Z \to \Z_{n_1} \times \cdots \times \Z_{n_k}$ with
$\sigma(x)  :=  (x \mod n_1,\ldots,x\mod n_k)$ for each $x \in \Z$.
\label{prop:crt}
\end{proposition}

\noindent
\textbf{Groups:}
We briefly recall basic concepts from group theory and permutation groups
(cf. see \cite{group-book}).
A \defn{group} $G$ is a pair $(S,\cdot)$, where $S$ is a set and $\cdot: (S 
\times S) \to S$ is a binary
operator satisfying: (i) associativity (i.e. $g_1 \cdot (g_2 \cdot g_3) = (g_1 
\cdot g_2) \cdot g_3$), (ii) the existence of a (unique) identity element $e 
\in S$ such that $g \cdot e = e \cdot g = g$ for all $g \in S$, and (iii) 
closure
under inverse (i.e. for each $g \in G$, there exists $g^{-1} \in G$ such that
$g \cdot g^{-1} = g^{-1} \cdot g = e$). When it is clear from the context, we 
will write $g \cdot g'$ as $gg'$. 
The \defn{order} 
$\Ord(G)$ of the group $G$ is defined to be the number $|S|$ of elements in
$G$.  This paper 
concerns only finite groups, i.e., groups $G$ with $\Ord(G) = |S| \in \N$. 
For each $n \in \N$, we define $g^n$ by
induction: (i) $g^0 = e$, and (ii) $g^n = g^{n-1} \cdot g$. 
The
\defn{order} $\Ord(g)$ of $g \in G$ is the least positive integer $n$
such that $g^n = e$. 

A \defn{subgroup} $H$ of $G = (S,\cdot)$ (denoted as $H \leq G$) is any group 
$(S',\cdot_H)$ such that $S' \subseteq S$ and $\cdot_H$ and $\cdot$ agree on 
$S'$. Lagrange's Theorem states that the order $\Ord(H)$ of $H$ divides
the order $\Ord(G)$ of $G$.
Given any subset $X \subseteq S$, the subgroup $\langle X \rangle$ of $G$ 
\defn{generated} by $X$ consists of those
elements of $G$ which can be expressed as a finite product of elements of
$X$ and their inverses. If $H = \langle X \rangle$, then $X$ is said to 
\defn{generate} $H$. A \defn{cyclic group} is a group generated by a singleton 
set $X = \{g\}$. 

An \defn{action} of a group $G = (S,\cdot)$ on a set $Y$ is a function
$\times: S \times Y \to Y$ such that for all $g,h \in S$ and $y \in Y$: (1) 
$(gh)\times y = g\times (h\times y)$, and (2) $e\times y = y$. 
The \defn{stabiliser of $x$ by $G$} is the subgroup 
$\Stab_G(x) := \{ g \in G : g \times x = x \}$ of $G$. If $G$ is understood,
$\Stab(x)$ will be used to denote $\Stab_G(x)$.
The \defn{($G$-)orbit}
containing $y$, denoted $Gy$, is the subset $\{ g\times y: g \in G\}$ of $Y$. 
The 
action $\times$ partitions the set $Y$ into $G$-orbits. When the meaning is 
clear, we shall omit mention of the operator $\times$, e.g, condition (2) above
becomes $ey = y$.
\smallskip


\noindent
\textbf{Permutation Groups.} A \defn{permutation} on $[n]$ is any bijection 
$\pi: [n] \to [n]$. The set of
all permutations on $[n]$ forms the \defn{($n$th) full symmetry group}
$\permgroup_n$ 
under functional composition. We shall use the notation $\Identity$ to
denote the identity element of each $\permgroup_n$. A word $w = a_0\ldots a_{k-1} \in 
[n]^*$ containing 
distinct elements of $[n]$ (i.e. $a_i \neq a_j$ if $i \neq j$) can be used
to denote the permutation that maps $a_i \mapsto a_{i+1 \mod k}$ for
each $i \in [0,k)$ and fixes other elements of $[n]$. In this case, $w$ is
    called a \defn{cycle} (more precisely, $k$-cycle or \emph{transpositions} in
    the case when $k=2$), which we will often write in the 
standard notation $(a_0,\ldots,a_{k-1})$ so as to avoid confusion. 
Observe that $w$ and 
$\rightshift(w)$ represent the same cycle $c$. We will however fix a particular
ordering to represent $c$ (e.g. the word provided as input to the orbit 
problem). For this reason, if $\vecV \in \Gamma^n$ for some alphabet $\Gamma$,
the notation $\vecV[c]$ is well-defined (see General Notations above), which
means projections of $\vecV$ onto elements with indices in $c$, e.g.,
if $\vecV = (1,1,1,0)$ and $c = (1,4,2)$, then $\vecV[c] = (1,0,1)$. Any 
permutation can be written as a composition of disjoint cycles
\cite{group-book}. Each subgroup $G = (S,\cdot)$ of $\permgroup_n$ acts on the set
$\Gamma^n$ (over any finite alphabet $\Gamma$) under the group action of
permuting indices, i.e., for each $\pi \in S$ and $\vecV = (a_1,\ldots,a_n) \in 
\Gamma^n$, we define $\pi\vecV := (a_{\pi^{-1}(1)},\ldots,a_{\pi^{-1}(n)})$. 
\smallskip

\OMIT{
For a cycle, $|c|$ denotes its length and $\Size{c}$ denotes its size.
For example, $c = (12,1,3)$ and $|c| = 3$, while $\Size{c} = 
(\floor{\log(12)} + 1) + (\floor{1} + 1) + (\floor{\log(3)} + 1)
= 4 + 1 + 3 = 8$. 
}






\noindent
\textbf{Complexity Analysis:} We will assume that permutations will be given 
in the input as a composition of disjoint cycles. It is easy to see that 
permutations can be converted back and forth in linear time from such 
representations and the representations of permutations as functions. The
size $\Size{n}$ of a number $n \in \N$ is defined to be the length of the binary
representation of $n$, which is $\floor{\log n}+1$. The size $\Size{c}$ of a 
cycle
$c = (a_1,\ldots,a_k)$ on $[n]$ is defined to be $\sum_{i=1}^k \Size{a_i}$
(in contrast, the length $|c|$ of $c$ is $k$).
For a permutation $\pi = c_1\cdots c_m$ where each $c_i$ is a cycle, the size
$\Size{\pi}$ of $\pi$ is defined to be $\sum_{i=1}^m \Size{c_i}$. We will use 
standard asymptotic notations from analysis of algorithms (big-O and little-o), 
cf.  \cite{Cormen}. We also use the standard $\sim$ notation: $f(n) \sim g(n)$ 
iff $\lim_{n\to\infty} f(n)/g(n) = 1$. We will use the standard RAM model that 
is 
commonly used when analysing the complexity of algorithms (cf. \cite{Cormen}).
In Section \ref{sec:reduce}, we will assume that
integer arithmetic takes constant time. Later in Section \ref{sec:bit},
we will use the \defn{bit complexity model} (cf. \cite{Cormen}), wherein the 
running time is measured in the number of bit operations.

\section{Reducing to solving a system of linear congruence equations}
\label{sec:reduce}
The main result of the paper is:
\begin{theorem}
There is a linear-time algorithm for solving the orbit problem when
the acting group is cyclic.
\label{th:onegen}
\end{theorem}
In this section, we will prove this theorem \emph{assuming constant-time
arithmetic operations}. In the next section, we will show that this theorem
still holds for the bit complexity model.


Before we proceed to the algorithm,
the following proposition shows why the naive algorithm that checks whether 
$g^i\vecV = \vecW$, for a given permutation $g \in \permgroup_n$ and for each 
$i\in [0,\Ord(g))$, actually runs in exponential time.
\begin{proposition}
There exists a sequence $\{G_i\}_{i=1}^{\infty}$ of cyclic groups $G_i = 
\langle g_i \rangle$ such that $\Ord(g_i)$ is exponential in the size
$\Size{g_i}$ of the permutation $g_i$.
\label{prop:exp}
\end{proposition}
\begin{proof}
Let $p_n$ denote the $n$th prime. The \defn{Prime Number Theorem} states that 
$p_n \sim n\log n$
(cf. \cite{HardyWright}). For each $i \in \Z_{> 0}$, we define a cycle 
$c_i$ of length $p_i$ by induction on $i$. For $i = 1$, let $c_1 = (1,2)$. 
Suppose that $c_{i-1} = (j,\ldots,k)$. In this case, we define 
$c_i$ to be the cycle $(k+1,\ldots,k+p_i)$. To define the sequence
$\{g_i\}_{i=1}^{\infty}$ of permutations, simply let $g_i = \Pi_{j=1}^i c_i$.
For example, we have $g_3 = (1,2)(3,4,5)(6,7,8,9,10)$. Since $c_i$'s are 
disjoint, the order $\Ord(g_i)$ of $g_i$ is the smallest positive integer $k$
such that $c_j^k = \Identity$ for all $j \in [i]$. If $S_j$ denotes the set of 
integers $k$ satisfying $c_j^k = \Identity$, then $\Ord(g_i)$ is precisely the smallest
positive integer in the set $\bigcap_{j=1}^i S_j$. It is easy to see that 
$S_j = p_j\Z$, which is the set of solutions to the linear congruence 
equation $x \equiv 0 \pmod{p_j}$. Therefore, by the Chinese Remainder Theorem
(cf. Propositon \ref{prop:crt}), the set $\bigcap_{j=1}^i S_j$ coincides
with the arithmetic progression $t_i\Z$ with $t_i := \prod_{j=1}^i p_j$. 
This implies that $\Ord(g_i) = t_i$. Now the number $t_i$ is also known as the 
\defn{$i$th primorial number} \cite{primorial} with $t_i \sim e^{(1+o(1))i \log 
i}$, which is a corollary of the Prime Number Theorem. On the other hand,
the size of $g_i$ is $\sum(i) := \sum_{j=1}^i p_i$, which is known to be
$\sim \frac{1}{2} i^2 \ln i$ (cf. \cite{BS96}). Therefore, $\Ord(g_i)$ is 
exponential in $\Size{g_i}$ as desired. \qed
\end{proof}

\begin{algorithm}
\caption{Reduction to system of modular arithmetic equations\label{algo:main}}
\begin{algorithmic}
\REQUIRE A permutation $g = c_1\cdots c_m \in \permgroup_n$, a finite alphabet
    $\Gamma$, and $\vecV,\vecW \in \Gamma^n$.
\ENSURE A system of modular arithmetic equations, which is satisfiable
iff $\exists i\in\N: g^i(\vecV) = \vecW$.
\STATE \emph{// First solve for each individual cycle}
\FORALL{$i=1,\ldots,m$} 
    \STATE Compute the length $|c_i|$ of the cycle $c_i$;
    \STATE Compute an ordered list $S_i' \subseteq [0,|c_i|)$ of numbers $r$ 
        with $c_i^r(\vecV[c_i]) = \vecW[c_i]$;
    \STATE \textbf{if} $S_i' = \emptyset$ \textbf{then} \textbf{return} $\false$ \textbf{end if}
    \STATE \textbf{if} $|S_i'| = 1$ \textbf{then} let $a_i$ be the member
        of $S_i$; $b_i := |c_i|$; \textbf{end if}
    \STATE \textbf{if} $|S_i'| > 1$ \textbf{then} $a_i := \min(S_i')$;
        $a_i' := \min(S_i'\setminus\{a_i\})$; $b_i := a_i' - a_i$; \textbf{end if}
\ENDFOR
\STATE \emph{// Now for each $i\in[1,m]$ we have a modular arithmetic
    equation $x \equiv a_i \pmod{b_i}$}
\STATE \textbf{return} YES iff there exists $x \in \N$ satisfying 
    $\bigwedge_{i=1}^m x \equiv a_i \pmod{b_i}$
\end{algorithmic}
\end{algorithm}

Our linear-time reduction that witnesses Theorem \ref{th:onegen} is given
in Algorithm \ref{algo:main}.
In this algorithm, the acting group is $G = \langle g \rangle$ with $g \in 
\permgroup_n$, expressed as a composition of disjoint cycles in a standard way, 
say,
$g = c_1c_2\cdots c_m$ where each $c_i$ is a cycle. Also part of the input is 
two strings $\vecV = v_1\ldots v_n,\vecW = w_1\ldots w_n \in \Gamma^n$ 
over a finite alphabet $\Gamma$. The orbit problem is to check whether
$f \vecV = \vecW$ for some $f \in G$, i.e., $f = g^r$ for some $r \in \Z$.
Since $c_i$'s are pairwise disjoint cycles, the question reduces to checking
if there exists $r \in \N$ such that
\[
    \forall i\in[1,m]: (c_i^r\vecV)[c_i] = \vecW[c_i] 
\]
In other words, for each $i \in [1,m]$, applying the action $c_i^r$ to $\vecV$ 
gives us $\vecW$ when restricted to the indices in $c_i$. To simplify 
notations in the above equation, we fix a letter $a \in \Gamma$ and, for
each $i \in [1,m]$, let the notation $\vecV_i$ (resp. $\vecW_i$) denote
the string $\vecV$ (resp. $\vecW$) in which all letters but those in positions 
$c_i$ are replaced by $a$. The equation above, therefore, amounts to 
\begin{equation}
    \forall i\in[1,m]: c_i^r\vecV_i = \vecW_i \tag{$*$}
\end{equation}
Essentially,
Algorithm \ref{algo:main} sequentially goes through each cycle $c_i$ and
computes the set $S_i$ of solutions $r$ to $c_i^r\vecV_i = \vecW_i$
as the set of solutions to the linear congruence equation 
$x \equiv a_i \pmod{b_i}$. Therefore, the set of solutions to (*) is precisely
the set of solutions to the system of congruence equations 
$\bigwedge_{i=1}^m x \equiv a_i \pmod{b_i}$, which by Proposition 
\ref{prop:sys-eqs} can be solved in linear time.

To compute $S_i$, we first prove a simple canonical form for $S_i$.
\begin{lemma}
For each $i=1,\ldots,m$, either $S_i = \emptyset$ or $S_i = a_i + b_i\Z$
for some $a_i \in [0,b_i)$ and $b_i \in (0,|c_i|]$ where $b_i$ divides 
$|c_i|$. 
    \OMIT{
    In the case
when $|S_i'| > 1$, we have $a_i = p_1$ and $b_i = p_2-p_1$, where $p_1 <
p_2$ are the smallest numbers in $S_i'$. Furthermore,
we may compute the pair $(a_i,b_i)$ of numbers in time $O(\Size{c_i})$.
}
\label{lm:normal}
\end{lemma}
\begin{proof}
    Let $G_i = \langle c_i \rangle$ be the group generated by $c_i$.
    Consider the
    stabiliser $H := \Stab(\vecV_i)$ of $\vecV_i$ by $G_i$. 
    Then, $H$ is a subgroup of $G_i$. Since $G_i$ is a cyclic group of order
    $|c_i|$, $H$ is a cyclic group generated by some element $h = c_i^k$, where 
    $k$ is the smallest integer in $(0,|c_i|]$ such that $c_i^k \in H$. 
    It is known that $k$ must be a divisor of $|c_i|$. This implies that the 
    orbit containing $\vecV_i$ consists of precisely $k$ elements $\vecV_i, 
    c_i\vecV_i, \ldots, c_i^{k-1}\vecV_i$. 
    \OMIT{
    Every subgroup $H'$
    of a cyclic group $G'$ generated by $a$ is cyclic and is generated by
    $a^k$, where $k$ is 
    the smallest positive integer such that 
    $a^k$ is in $H'$ (e.g. see \cite[Theorem 7.16]{Hungerford}). Therefore, $H$ 
    is a
    subgroup of $G_i$ and is generated by some element $h = c_i^k$, where $k$
    is the smallest integer in $(0,|c_i|]$ such that $c_i^k \in H$. In fact, 
    $k$ must divide $|c_i|$. For,
    otherwise, $k < |c_i|$ and we let $q$ be the smallest integer such that
    $qk > |c_i|$. That is, $k \leq (q-1)k < |c_i|$. Thus, we have 
    $0 < qk - |c_i| < k$ contradicting that $k$ is the smallest integer in
    $(0,|c_i|]$ such that $c_i^k \in H$. Therefore, it follows that $k$ 
    divides $|c_i|$. This implies that the ($G_i$-)orbit of $\vecV_i$ 
    consists of precisely $k$ elements 
    $\vecV_i, c_i\vecV_i, \ldots, c_i^{k-1}\vecV_i$. 
}

    Suppose that $S_i \neq \emptyset$. Let $s$ be the smallest nonnegative
    integer in $S_i$, i.e., $c_i^s \vecV_i = \vecW_i$. Then, $s \in [0,k)$.
    We claim that $S_i = s + k\Z$. We have $s+k\Z \subseteq S_i$ since 
    $c_i^{s+kn}\vecV_i = c_i^s(c_i^{kn}\vecV_i) = c_i^s\vecV_i$. Conversely,
    if $t \in S_i$, then $c^t\vecV_i = c^p\vecV_i$, where $p$ is the smallest
    nonnegative integer such that $p \equiv t \pmod{k}$. Then, it must be
    the case that $p = s$ since $\vecV_i,c_i\vecV_i,\ldots,c_i^{k-1}\vecV_i$
    are all different. Thus, it follows that $S_i \subseteq s + k\Z$. Letting 
    $a_i = s$ and $b_i = k$ completes the proof. 
    \qed
\end{proof}
In view of Lemma \ref{lm:normal}, it suffices to show how to determine if
$S_i \neq \emptyset$ and, if so, compute $a_i$ and
$b_i$ in time $O(\Size{c_i})$. We first compute the 
length $|c_i|$ of the cycle $c_i$, which can be done in time $O(\Size{c_i})$.
[This is the same as how to compute the length of a list.] We proceed by
computing representatives $S_i' \subseteq [0,|c_i|)$ for $S_i$. This suffices
to compute $a_i$ and $b_i$ since:
    \begin{itemize}
        \item If $S_i' = \emptyset$, then $S_i = \emptyset$.
        \item If $S_i' = \{a\}$, then $S_i = a + |c_i|\Z$.
        \item If $|S_i'| > 1$, then $S_i = a_i + b_i\Z$, whenever $a_i$ and
            $a_i + b_i$ are the two smallest numbers in $S_i'$.
    \end{itemize}
This case-by-case treatment is reflected in Line 3--Line 5 within the for-loop
in Algorithm \ref{algo:main}.
To compute $S_i'$, we collect a subset of numbers $h \in [0,|c_i|)$ such that 
$c_i^h\vecV_i = \vecW_i$. A quadratic algorithm for this is easy to
come up with: sequentially go through $h \in [0,|c_i|)$ while computing the
current $c_i^h$, and save $h$ if $c_i^h\vecV_i = \vecW_i$ holds. 
One way to obtain a linear-time algorithm is to reduce our problem to
the 
\defn{string searching problem}: given a ``text'' $T \in \Sigma^*$ (over some 
finite alphabet $\Sigma$) and a ``pattern'' $P \in \Sigma^*$, find all 
positions $i$ in $T$ such that $T[i,i+|P|] = P$. This problem is solvable in 
time $O(|T|+|P|)$ by Knuth-Morris-Pratt (KMP) algorithm (e.g. see 
\cite{Cormen}).


\OMIT{
The number $|c_i|$ is stored in binary counter and can be computed by counting
upwards from 0 and incrementing by 1 as we go through the elements in $c_i$.
Although a single increment by 1 might take $O(|c_i|)$ bit operations
in the worst case (since we have to propagate the carry bit), it is known (e.g. 
see \cite[Chapter 17, p. 454]{Cormen}) 
that the entire sequence of operations actually takes time $O(|c_i|)$.
Therefore, accounting for all the cycles, this step takes
$\sum_{i=1}^m O(\Size{c_i}) = O(\sum_{i=1}^m \Size{c_i}) = O(\Size{g})$, which
is linear in the input size.
}

\OMIT{
\smallskip
\noindent
\textbf{Step 2: Computing representatives $\boldsymbol{S_i' \subseteq [0,|c_i|)}$ for
$S_i$.}
During this step, we collect a subset of numbers $h \in [0,|c_i|)$ such that 
$c_i^h(\vecV[c_i]) = \vecW[c_i]$. A quadratic algorithm for this is easy to
come up with: sequentially go through $h \in [0,|c_i|)$ while computing the
current $c_i^h$, and save $h$ if $c_i^h(\vecV[c_i]) = \vecW[c_i]$ holds. 
One way to obtain a linear-time algorithm is to reduce our problem to
the 
\defn{string searching problem}: given a ``text'' $T \in \Sigma^*$ (over some 
finite alphabet $\Sigma$) and a ``pattern'' $P \in \Sigma^*$, find all 
positions $i$ in $T$ such that $T[i,i+|P|] = P$. This problem is solvable in 
linear-time by Knuth-Morris-Pratt (KMP) algorithm (e.g. see \cite{Cormen}). 
}

We now show how to reduce our problem to the string searching problem in linear
time. We will also use the following running example to illustrate the 
reduction: $c = (6,5,7,3,2,1)$,
$\vecV = \underline{010}0\underline{011}11$, and
$\vecW = \underline{101}1\underline{100}01$, where the positions in
$\vecV$ and $\vecW$ that are modified by $c$ are underlined.
Below, we will work with the equivalent equation $(c_i^r\vecV)[c_i] = 
\vecW[c_i]$ (i.e. instead of $c_i^r\vecV_i = \vecW_i$).
Suppose that $c := c_i = (j_1,\ldots,j_k)$. We have $\vecV[c]
= v_{j_1}\ldots v_{j_k}$ and $\vecW[c] = w_{j_1}\ldots w_{j_k}$. 
\begin{lemma}
$(c\vecV)[c] = \rightshift(\vecV[c])$.
\label{lm:rs}
\end{lemma}

\noindent
In other words, if $\Dom(c) = \{j_1,\ldots,j_k\}$, the effect of $c$ on 
$\vecV$ when restricted to $\Dom(c)$ coincides with applying a cyclical right
shift on the string $\vecV[c]$. Following our running example, it is easy
to check that $\vecV[c] = 101010$ and $(c\vecV)[c] = \rightshift(\vecV[c]) = 
010101$.
\begin{proof}[of Lemma \ref{lm:rs}]
Let $\vecU = u_1\ldots u_k := (c\vecV)[c]$ and $\vecU = u_1'\ldots u_k' :=
\rightshift(\vecV[c])$. It suffices to show that $u_t = u_t'$ for all $t \in
\Z_k$. By definition of $\rightshift$, it follows that $u_t' = v_{j_{t-1}}$.
Now suppose that $\vecV' = v_1'\ldots v_n' := c\vecV$. Then 
\[
    v_j' := \left\{ \begin{array}{cc}
                        v_j & \quad \text{if $j \notin \Dom(c)$} \\
                        v_{j'} & \quad \text{if $j \in \Dom(c)$ and,
                                                for some $t \in \Z_k$,
                                                $j = j_{t+1}$ and 
                                                $j' = j_t$.}
                    \end{array}
           \right.
\]
So, for $t \in \Dom(c)$, we have $u_t = ((c\vecV)[c])[t] = (\vecV'[c])[t] = v_{j_t}'
= v_{j_{t-1}}$. This proves that $u_t = u_t'$. \qed
\end{proof}
\begin{lemma}
For each $r \in \N$, we have $(c^r\vecV)[c] = \rightshift^r(\vecV[c])$.
\label{lm:rs:induction}
\end{lemma}
\begin{proof}
This lemma can be proved by induction on $r \in \N$. The base case $r = 0$
is vacuous. For the induction case, we assume the induction hypothesis: 
$c^{r-1}\vecV = \rightshift^{r-1}(\vecV[c])$. It follows that
\[
    (c^r\vecV)[c] = (c(c^{r-1}\vecV))[c] = \rightshift((c^{r-1}\vecV)[c])
    = \rightshift(\rightshift^{r-1}(\vecV[c])) = \rightshift^r(\vecV[c]).
\]
The third equality is by Lemma \ref{lm:rs}, while the fourth equality is by
the induction hypothesis. This completes the proof. \qed 
\end{proof}

\OMIT{
Lemma \ref{lm:rs:induction} implies that the set $S := S_i \subseteq \N$ of 
solutions $r$ to the equation $(c_i^r\vecV)[c_i] = \vecW[c_i]$ is a finite 
union of arithmetic
progressions of the form $a+k\Z$, where $k = |c_i|$ and $a \in [0,k)$. This is 
simply because $\rightshift^{r+k}(\vecV[c_i]) = \rightshift^r(\vecV[c_i])$. We 
will finitely represent $S$ by the offsets $a$'s and the unique period $k$ in 
these arithmetic progressions.
}

Define the text $T := 
\vecV[c]\vecV[c]$ and the pattern $P := \vecW[c]$. Observe that, for each 
$r \in [k]$, $P$ is matched at position $r$ in $T$ iff 
$\rightshift^{r-1}(\vecV[c])
= \vecW[c]$. Therefore, after running the KMP algorithm with the solution set
$S'$, the set $S_i'$ will be $\{ r -1 : r \in S'\}$. Observe that this
step takes time $O(|T|+|P|) = O(\Size{c})$.
\begin{example}
Continuing with our running example, it follows that
$T = \vecV[c]\vecV[c] = 101010101010$ and $P = \vecW[c] = 010101$. We see
that $P$ matches $T$ at positions $S' = \{2,4,6\}$. This implies that 
$S_i' = \{1,3,5\}$ and so the set $S_i$ of solutions $r \in \Z$ to the 
equation $c^r\vecV_i = \vecW_i$ is $1 + 2\Z$.
$\blacksquare$
\end{example}

\noindent
\textbf{Summing up.} To sum up, the time spent computing the linear
congruence equation $x \equiv a_i \pmod{b_i}$ for each $i\in[1,m]$ is
$O(\Size{c_i})$. Therefore, our reduction runs in time 
$O(\sum_{i=1}^m \Size{c_i}) = O(\Size{g})$, which is linear in input size.
Therefore, invoking Proposition \ref{prop:sys-eqs} on the resulting
system of linear congruence equations, we obtain the set of solutions to
(*) in linear time.
\begin{example}
Let us continue with our running example. Let 
\[
    g_1 := c(4,8) = (6,5,7,3,2,1)(4,8), \quad g_2 := c(4,8,9) = 
        (6,5,7,3,2,1)(4,8,9).
\]
Then, running Algorithm \ref{algo:main} on $g_1$ yields the system 
$x \equiv 1 \pmod{2} \wedge x\equiv 1\pmod{2}$, which is equivalent to
$x \equiv 1 \pmod{2}$. Running Algorithm \ref{algo:main} on $g_2$ yields
the system $x \equiv 1 \pmod{2} \wedge x \equiv 1 \pmod{3}$. Both systems
are solvable. $\blacksquare$
\end{example}

\section{Making do with linearly many bit operations}
\label{sec:bit}
Thus far, we have assumed that arithmetic operations take constant time.
In this section, since Algorithm \ref{algo:sys-eqs} makes a substantial
use of basic arithmetic operations, we will revisit this assumption. It turns
out that, although our reduction (Algorithm \ref{algo:main}) to solving a 
system of linear congruence 
equations runs in linear time in the bit complexity model, the algorithm for
solving the system of equations (Algorithm \ref{algo:sys-eqs}) uses at least a
cubic number of bit-arithmetic operations. The main results in this section are 
two-fold: (1) on inputs
given by our reduction, Algorithm \ref{algo:sys-eqs} runs in sublinear time
(more precisely, $O(\log^5 n)$)
\emph{on average} in the bit complexity model, and (2) there exists another 
algorithm for solving a system of linear congruence equations (with numbers
in the input represented in unary) that runs in linear time in the bit 
complexity model in the worst case.

We begin with two lemmas that provide the running time of Algorithm
\ref{algo:main} and Algorithm \ref{algo:sys-eqs} in the bit complexity model.
\begin{lemma}
Algorithm \ref{algo:main} runs in linear time in the bit complexity model.
\label{lm:bit:linear-main}
\end{lemma}
\begin{proof}
On $i$th iteration, the number $|c_i|$ is stored in binary counter and can be 
computed by counting upwards from 0 and incrementing by 1 as we go through the 
elements in $c_i$.
Although a single increment by 1 might take $O(|c_i|)$ bit operations
in the worst case (since we have to propagate the carry bit), it is known (e.g. 
see \cite[Chapter 17, p. 454]{Cormen}) 
that the entire sequence of operations actually takes time $O(|c_i|)$.
Finally, since addition and substraction of two numbers can easily be performed 
in $O(\beta)$ time on numbers that use at most $\beta$ bits,
the operation $b_i := a_i' - a_i$ on the last line of the iteration
takes at most $O(\log |c_i|)$ time. Therefore, accounting for all the cycles, 
the algorithm takes
$\sum_{i=1}^m O(\Size{c_i}) = O(\sum_{i=1}^m \Size{c_i}) = O(\Size{g})$, which
is linear in the input size. \qed
\end{proof}
\begin{lemma}
On an input $\bigwedge_{i=1}^m x \equiv a_i \pmod{b_i}$ with $N = \max\{ b_i :
i\in[1,m]\}$, Algorithm \ref{algo:sys-eqs} uses at most $m\log N$ bits to 
store any numeric variables. Furthermore, the algorithm runs in time 
$O(m^3 \log^2 N)$ in the bit complexity model.
\label{lm:superlinear}
\end{lemma}
\begin{proof}
On $i$th iteration, the number of bits used to store $a$ and $b$ grow by
at most $\log b_i$. On the other hand, the invariant that $a', b' \in [0,b_i)$
is always maintained on the $i$th iteration and so they only need at most
$\log N$ bits to represent throughout the algorithm. Hence, the algorithm uses 
$M = O(m\log N)$ bits to store $a$, $b$, $a'$, and $b'$. Extended Euclidean 
Algorithm runs in time $O(M^2)$ on inputs where each number uses at most 
$M$ bits (cf. \cite[Problem 31-2]{Cormen}), which also bounds the time it takes
on each iteration. Therefore, the algorithm takes at most $O(mM^2) = O(m^3
\log^2 N)$ in the bit complexity model. \qed
\end{proof}

We now provide an average case analysis of the running time of Algorithm
\ref{algo:sys-eqs} on system of linear congruence equations given by our
reduction. The input to the orbit problem over cyclic groups includes a 
permutation $g \in S_n$ and two vectors $\vecV,\vecW \in \Gamma^n$. We briefly
recall the setting of average-case analysis (cf. \cite{Sedgewick}). Let
$\Pi_N$ be the set of all inputs to the algorithm of size $N$. Likewise, let 
$\Sigma_N$ be the sum of the \emph{costs} (i.e. running time) of the algorithm
on \emph{all} inputs of size $N$. Hence, if $\Pi_{N,k}$ is the cost of the
algorithm on input of size $N$ with running time $k$, then $\Sigma_N = \sum_k 
k\Pi_{N,k}$.
The \defn{average case complexity of the algorithm} is defined to be
$\Sigma_N/\Pi_N$.
\begin{theorem}
The expected running time of Algorithm \ref{algo:sys-eqs} in the bit 
complexity model on inputs provided by Algorithm \ref{algo:main} is 
$O(\log^5 n)$.
\label{th:average}
\end{theorem}
\begin{proof}
The size of a single permutation $g \in S_n$ is $O(n)$ and additionally 
$\Pi_n = |S_n| = n!$. Suppose that $g$ has $k$ cycles (say, $g = c_1\cdots 
c_k$). Then, 
Algorithm \ref{algo:main}
produces a system of equations $\bigwedge_{i=1}^k x \equiv a_i \pmod{b_i}$,
where $a_i, b_i \in [0,|c_i|)$. By Lemma \ref{lm:superlinear}, Algorithm
\ref{algo:sys-eqs} takes $O(k^3 \log^2 n)$ time in the bit complexity model,
since 
$N := \max\{ b_i : i \in [1,m]\} \leq n$. 
In addition, the number of permutations in $S_n$ with
$k$ cycles is precisely the definition of the \emph{unsigned Stirling number of 
the first kind} $\StirlingFirst{n}{k}$. Therefore, we have 
$\Sigma_n = O\left(\sum_{k=1}^n (k^3\log^2 n) \StirlingFirst{n}{k}\right) = 
O\left(\log^2 n \sum_{k=1}^n k^3 \StirlingFirst{n}{k}\right)$. Therefore, it 
suffices
to show that $\frac{1}{n!}\sum_{k=1}^n k^3 \StirlingFirst{n}{k} \sim c\log^3 n$
for a constant $c$. The proof can be found in the appendix. \qed
\end{proof}

Finally, we will now give our final main result of this section.
\begin{theorem}
There exists a linear-time algorithm in the bit complexity model for solving a 
system of linear congruence 
equations when the input numbers are represented in unary.
\label{th:unary}
\end{theorem}

\noindent
We now provide an algorithm that witnesses the above theorem. 
Let $\bigwedge_{i=1}^m x \equiv a_i \pmod{b_i}$ be the given system of 
equations. With unary representation of numbers, the size $N_i$ of the
equation $x \equiv a_i \pmod{b_i}$ is $a_i + b_i$. We use $n$ to denote the
total number of bits in the system of equations. Initially, we compute 
a binary representation of all the numbers $a_i$'s, $b_i$'s, and $n$ as in the 
proof of Lemma \ref{lm:bit:linear-main}, which takes linear time.
Next we factorise all the numbers $b_i$ into a product of distinct prime 
powers $p_{j_{i1}}^{e_{i1}}\cdots p_{j_{it_i}}^{e_{it_i}}$, where $p_j$
stands for the $j$th prime and all $e_{ij}$'s are positive integers.
%
This can be done in time $O(\sqrt{N_i} \log^2 N_i)$.
To obtain this time bound, we can use any
\emph{unconditional}\footnote{This means that the bound does not depend on
any number-theoretic assumptions.} deterministic
factorisation methods like
Strassen's algorithm, whose complexity was shown in \cite{BGS07} 
(cf. also see \cite{CH12}) to be $O(f(N^{1/4}\log N))$ for factoring a number
$N$, where $f(M)$ is the number of bit operations required
to multiply two numbers with $M$ bits. The standard (high-school) multiplication
algorithm runs in quadratic time giving us $f(M) = O(M^2)$, which
suffices for our purposes. This shows that Strassen's algorithm runs in time 
$O(N^{1/2}\log^2 N)$. 
[In practice,
do factoring using the general number field sieve 
(cf. \cite{Cormen}),
which performs extremely well in practice, though its complexity
requires some unproven number-theoretic assumptions.]

Next, following Chinese Remainder Theorem (CRT), compute $z_{ij} := a_i 
\mod p_{ij}^{e_{ij}}$ for 
each 
$j \in [1,t_i]$. Let us analyse the time complexity for performing this.
Each $z_{ij}$ can be computed by a standard algorithm (e.g. see \cite{Cormen})
in time quadratic in the number
of bits used to represent $a_i$ and $p_{ij}^{e_{ij}}$. Since each of these 
numbers use at most $\log N_i$ bits, each $z_i$ can be computed in 
time $O(\log^2 N_i)$, which is $o(N_i)$. In 
addition, since $e_{ij} > 1$ for each $j \in [1,t_i]$, it follows that 
$t_i = O(\log N_i)$. This means that the total time it takes to compute 
$\{ z_{ij} : j\in [1,t_i]\}$ is $O(\log^3 N_i)$, which is also $o(N_i)$. So, 
computing this for all $i\in[1,m]$ takes time $O(\sum_{i=1}^m \log^3 N_i)$, 
which is at most linear in the input size.

In summary, for each $i\in[1,m]$, we obtained the following system of equations,
which is equivalent to $x \equiv a_i\pmod{b_i}$ by CRT:
\begin{equation}
    x \equiv z_{i1} \pmod{p_{i1}^{e_{i1}}} \quad \wedge \quad \cdots\cdots
    \quad
    \wedge
    \quad
    x \equiv z_{it_i} \pmod{p_{it_i}^{e_{it_i}}} \tag{$E_i$}
\end{equation}
The final step is to determine if there exists a number $x \in \N$ that 
satisfies \emph{each} ($E_i$), for all $i\in[1,m]$. Loosely, we will go through
all the equations and make sure that there is no conflict between any
two equations whose periods are powers of the same prime number, i.e., $x 
\equiv a \pmod{b}$ and $x \equiv a'\pmod{b'}$ such that $b = p^i$ and
$b' = p^{i'}$ for some prime $p$ and $i,i' \in \Z_{>0}$. In order to achieve
this in linear-time in the bit complexity model, one has to store these 
equations in the memory (in the form of lookup tables) and carefully perform
the lookup operations while looking for a conflict. To this end, we first 
compute
$p_{\max} = \max\{p_{ij} : i \in [1,m], j \in [1,t_i]\}$ and
$e_{\max} = \max\{e_{ij} : i \in [1,m], j \in [1,t_j]\}$.
\begin{lemma}
$p_{\max}$ and $e_{\max}$ can be computed using $O(n)$ many bit operations.
\label{lm:pmax}
\end{lemma}
\begin{proof}
The algorithm for computing $p_{\max}$ and $e_{\max}$ is a slight modification 
of the standard algorithm that computes the maximum number in a list, which
sequentially goes through the list $n_1,\ldots,n_m$ while keeping the maximum 
number $n_{\max}$ 
in the sublist explored so far. To ensure linear-time complexity, we have to 
make sure that when comparing the values of $n_i$ and $n_{\max}$, we explore
at most $n_i$ bits of $n_{\max}$ (since $n_{\max}$ is possibly much larger than
$n_i$). 
%
This is easily achievable by assuming 
binary representation of these numbers \emph{without redundant leading 0s},
e.g., the number 5 will be represented as 101, not 0101 or 00000101. That 
way, we will only need to inspect $\log(n_i)$ bits from $n_{\max}$ on the
$i$th iteration, which will give a total running time of $O(\sum_{i=1}^m
\log(n_i))$, which is linear in input size. \qed
\end{proof}

Next, keep one 1-dimensional array $A$ and one 
2-dimensional array $B$:
\[
A[1,\ldots,p_{\max}] \qquad \qquad B[1,\ldots,p_{\max}][1,\ldots,e_{\max}].
\]
$A[k]$ and $B[k][e]$ will not be defined when $k$ is not a prime number.
We will use $A[k]$ as a flag indicating whether some equation of the form
$x \equiv z \pmod{k^e}$ has been visited, in which case $A[k]$ will contain
$(z,e)$. In this case, we will use $B[k][e']$ (with $e' \leq e$) to store
the value of $z \mod k^{e'}$.

We now elaborate how $A$ and $B$ are used when iterating over the equations
in the system. Sequentially go through 
each system ($E_i$) of equations. For each $i\in[1,m]$,  sequentially go 
through each equation $x \equiv z_{ij} \pmod{p_{ij}^{e_{ij}}}$, for each
$j \in [1,t_i]$, and check if $A[p_{ij}]$ is defined. If it is not defined,
set $A[p_{ij}] := (z_{ij},e_{ij})$ and compute $B[p_{ij}][l] = z_{ij} \mod
p^l$ for each $l \in [1,e_{ij}]$. If it is defined (say,
$A[p_{ij}] = (z,e)$), then we analyse the constraints $x \equiv z 
\pmod{p_{ij}^e}$ and $x \equiv z_{ij}\pmod{p_{ij}^{e_{ij}}}$ simultaneously.
We compare $e$ and $e_{ij}$ resulting in three cases:
\begin{description}
\item[Case 1.] $e = e_{ij}$. In this case, make sure that 
$z = z_{ij}$ otherwise the two equations (and, hence, the entire system) cannot 
be satisfied simultaneously.
\item[Case 2.] $e < e_{ij}$. In this case, make sure that $z_{ij} \equiv z 
\pmod{p_{ij}^e}$ (otherwise, unsatisfiable) and assign $A[p_{ij}] := 
(z_{ij},e_{ij})$. For each $l \in [1,e_{ij}]$, update $B[p_{ij}][l] := z_{ij}
\mod p_{ij}^l$.
\item[Case 3.] $e > e_{ij}$. In this case, make sure that $z_{ij} \equiv z
\pmod{p_{ij}^{e_{ij}}}$ (otherwise, unsatisfiable).
\end{description}
We now analyse the running time of this final step (i.e. when scanning through
the subsystem ($E_i$)). To this end, we measure the time it 
takes to process each
equation $x \equiv z_{ij} \pmod{p_{ij}^{e_{ij}}}$. There are two cases, which
we will analyse in turn. 
\smallskip

\noindent
\textbf{(Case I):} when $A[p_{ij}]$ is not defined. In this case, 
setting 
$A[p_{ij}]$ takes constant time, while setting $B[p_{ij}][l]$ for all
$l \in [1,e_{ij}]$ takes $O(e_{ij} \times (\log z_{ij} + 
\log p_{ij}^{e_{ij}})^2)$ since computing $a \mod b$ can be done in time 
quadratic in $\log(a)+\log(b)$. Since $e_{ij} \leq \log N_i$ and
$z_{ij},p_{ij} \leq N_i$, this
expression can be simplified to $O(\log N_i \times \log^2(z_{ij}N_ip_{ij}))
= O(\log^3 N_i)$.
\smallskip

\noindent
\textbf{(Case II): } when $A[p_{ij}]$ is already defined, e.g.,
$A[p_{ij}] = (z,e)$. In this case, we will compare the values of $e$ and 
$e_{ij}$. To ensure linear-time complexity, we will make sure that at most 
$\log(e_{ij})$ bits from $e$ are read by using the trick from the proof of Lemma
\ref{lm:pmax}. For Case 1, we will need extra $O(\log z_{ij}) = O(\log N_i)$ 
time steps.
For Case 2, we have $0 \leq z \leq p^{e_{ij}}$ and computing
$z_{ij} \mod p_{ij}^e$ can be done in time $O(\log^2 N_i)$ as before. Updating
$B[p_{ij}][l]$ for all $l \in [1,e_{ij}]$ takes $O(\log^3 N_i)$ as in the
previous paragraph. For Case 3, since $e > e_{ij}$, we may access the value of 
$z \mod p_{ij}^{e_{ij}}$ from $B[p_{ij}][e_{ij}]$ in constant time and
compare this with the value of $z_{ij}$. Since $z \in [0,p_{ij}^{e_{ij}})$,
this takes time $O(\log N_i)$.

In summary, either case takes time at most $O(\log^3 N_i)$. Therefore, 
accounting for the entire subsystem ($E_i$), the algorithm incurs
$O(\sum_{j=1}^{t_i} \log^3 N_i) = O(\log^4 N_i)$ time steps. Hence, accounting
for \emph{all} of the subsystems $E_i$ ($i\in[1,m]$) the
algorithm takes time $O(\sum_{i=1}^m \log^4 N_i)$, which is linear in the
size of the input. This completes the proof of Theorem \ref{th:unary}.
\begin{remark}
The purpose of the 2-dimensional array $B$ above is to avoid superlinear time
complexity for Case 3. We can imagine a system of linear equations 
$\bigwedge_{i=1}^m x \equiv a_i \pmod{b_i}$, where $a_1$ and $b_1$ are 
substantially larger than the other $a_i$'s and $b_i$'s ($i \in [2,m]$). In 
this case, without
the lookup table $B$, checking whether $a_i \equiv a_1 \pmod{b_i}$ in Case 3
will require the algorithm to inspect the entire value of $a_1$, which prevents
us from bounding the time complexity in terms of $a_i$ and will yield a 
superlinear time complexity for our algorithm.
\end{remark}

\OMIT{
handled in the same way. In summary, if this procedure successfully reaches the 
end of the iteration, we obtain an array $A$ containing a system of equations
\begin{equation*}
    x \equiv a_1 \pmod{p_1^{e_1}} \qquad \cdots\cdots\cdots\cdots \qquad
    x \equiv a_s \pmod{p_s^{e_s}}
\end{equation*}
such that $p_i$ and $p_j$ ($1 \leq i \neq j \leq s$) are distinct prime
numbers. This allows us to apply CRT and deduce that the system of equations
is soluble. 
To ensure linear-time complexity, we have to make sure that
when exploring the constraint $x \equiv z_{ij} \pmod{p_{ij}^{e_{ij}}}$
and if $A[p_{ij}] = (z,e)$ is defined (which can be accessed in constant time in
RAM model), the algorithm does not explore the entire content of $z$ and
$e$ (since their size might be bigger than $z_{ij}$ and $e_{ij}$). To achieve
this, we will represent these numbers \emph{without redundant leading 0s},
e.g., the number 5 will be represented as 101, not 0101 or 00000101. That 
way, we will only need to inspect $\log(z_{ij})$ (resp. $\log(e_{ij})$) bits
from $z$ and $e$.
}

\section{Future work}
\label{sec:conc}
Since an algorithm for the orbit problem will be invoked many times during an 
explicit-state model checking (in the worst case once each time a new state in 
the transition system is visited; cf. \cite{WD10}), we believe that it is 
important to further identify efficiently solvable (preferably, in linear-time) 
subcases of the orbit problem. As mentioned in the Introduction, there are known
classes of permutations groups whose orbit problem is polynomial-time solvable
(e.g. $\Gamma_d$ which contains solvable groups). We propose the question of
further identifying other classes of permutations groups whose orbit problem is
solvable in linear time. 
\bibliographystyle{spmpsci}      

\bibliography{references}

\appendix

\section{Completing proof of Theorem \ref{th:average}} 

Let $\StirlingFirst{n}{k}$ denote the unsigned Stirling number of
the first kind, and $\binom{n}{k}$ denote $n$ choose $k$. The harmonic number 
$H_n$ is defined as
$$
H_n = \sum_{k=1}^n \frac{1}{k}.
$$
In general, for an integer $s \ge 1$, the generalized harmonic number of order $s$ is defined as
$$
H_n^{(s)} = \sum_{k=1}^n \frac{1}{k^s}.
$$
It is known that 
$$
\frac{1}{n!} \sum_{k=1}^n k \StirlingFirst{n}{k} = H_n.
$$

Define 
$$
f(n) = \frac{1}{n!} \sum_{k=1}^n k^2 \StirlingFirst{n}{k}, \quad 
g(n) = \frac{1}{n!} \sum_{k=1}^n k^3 \StirlingFirst{n}{k}.
$$
It is known that
$$
\sum_{k=m}^{n} \StirlingFirst{n}{k} \binom{k}{m} = \StirlingFirst{n+1}{m+1},
$$
(see \cite{stirling}).
In particular, we have
\begin{eqnarray*}
\StirlingFirst{n+1}{3} & = & \sum_{k=2}^{n} \StirlingFirst{n}{k} \binom{k}{2} = \sum_{k=1}^{n} \StirlingFirst{n}{k} \frac{k(k-1)}{2} = \frac{1}{2} n! (f(n) - H_n), \\
\StirlingFirst{n+1}{4} & = & \sum_{k=3}^{n} \StirlingFirst{n}{k} \binom{k}{3} = \sum_{k=1}^{n} \StirlingFirst{n}{k} \frac{k(k-1)(k-2)}{6} = \frac{1}{6} n! (g(n) - 3f(n) + 2H_n). 
\end{eqnarray*}
That is, we have
$$
f(n) = \frac{2}{n!} \StirlingFirst{n+1}{3} + H_n, \text{ and}
$$
$$
g(n) = \frac{6}{n!} \StirlingFirst{n+1}{4} + 3f(n) - 2H_n = \frac{6}{n!} \StirlingFirst{n+1}{4} + \frac{6}{n!} \StirlingFirst{n+1}{3} + H_n. 
$$
It is known (cf. page 217 of \cite{comtet74}) 
that
$$
\frac{1}{n!} \StirlingFirst{n+1}{3} = \frac{1}{2} (H_n^2 - H_n^{(2)})
$$
$$
\frac{1}{n!} \StirlingFirst{n+1}{4} = \frac{1}{6} (H_n^3 - 3 H_n H_n^{(2)} + 2 H_n^{(3)}). 
$$
\OMIT{
(I found this on \url{http://en.wikipedia.org/wiki/Stirling_numbers_of_the_first_kind#Relation_to_harmonic_numbers}, but this should be double checked by using a more reliable publication.)
}
It is also known that $H_n = \gamma + \ln n$, 
$\lim_{n \rightarrow \infty} H_n^{(2)} = \zeta(2) = \frac{\pi^2}{6}$ and $\lim_{n \rightarrow \infty} H_n^{(3)} = \zeta(3) \approx 1.202$, where $\gamma \approx 0.577$ is Euler's constant and $\zeta(s) = \sum_{k=1}^{\infty} \frac{1}{k^s}$ is the Riemann zeta function. Hence we obtain
$$
f(n) = H_n^2 - H_n^{(n)} + H_n \sim \ln^2 n.
$$
Putting all together, we obtain
$$
g(n) = \frac{6}{n!} \StirlingFirst{n+1}{4} + 3f(n) - 2H_n \sim \ln^3 n.
$$
\OMIT{
(A precise formula for $g(n)$ can be obtained but is unnecessary. In the same fashion one can show that $f(n) \sim \log^2 n$.)
}


\end{document}